\documentclass{ejm2}
\usepackage{latexsym,amsfonts,amsmath,cyracc,verbatim}
\usepackage{graphicx}

\title[Homogenization of a locally-periodic medium]{Homogenization of a locally-periodic medium with areas of low and high diffusivity}

\author[Tycho L. van Noorden and Adrian Muntean]{T.\ns L.\ns V\ls A\ls N\ns N\ls O\ls O\ls R\ls D\ls E\ls N$\,^1$\ns
\and A.\ns M\ls U\ls N\ls T\ls E\ls A\ls N$\,^{1,2}$}

\affiliation{$^1\,$Department of Mathematics and Computer Science,
Technische Universiteit Eindhoven, P.O. Box 513, 5600 MB Eindhoven,
The Netherlands\\
$^2\,$Institute of Complex Molecular Systems (ICMS), Technische
Universiteit Eindhoven, P.O. Box 513, 5600 MB Eindhoven, The
Netherlands}

\begin{document}
\maketitle

\begin{abstract}
We aim at understanding transport in porous materials including regions with both high and low diffusivities. For such scenarios, the transport becomes structured (here: {\em micro-macro}).
The geometry we have in mind includes regions of low diffusivity arranged in a locally-periodic fashion.
We choose a prototypical advection-diffusion system (of minimal size),
discuss its formal homogenization (the heterogenous medium being now assumed to be made of zones with
 circular areas of low diffusivity of $x$-varying sizes), and prove the weak solvability of the limit two-scale reaction-diffusion model.
  A special feature of our analysis is that most of the basic estimates (positivity, $L^\infty$-bounds, uniqueness, energy inequality)
   are obtained in $x$-dependent Bochner spaces.
\end{abstract}


\newtheorem{theorem}{Theorem}[section]
\newdefinition{remark}[theorem]{Remark}
\newdefinition{assumption}{Assumption}
\newdefinition{lemma}[theorem]{Lemma}
\newdefinition{claim}[theorem]{Claim}
\newdefinition{definition}[theorem]{Definition}
\newdefinition{proposition}[theorem]{Proposition}

\section{Introduction}

We consider transport in heterogeneous media presenting regions with high and low diffusivities. Examples of such media are concrete and scavenger packaging materials.
 For the scenario we have in mind, the old classical idea to replace the heterogeneous medium by a homogeneous equivalent
  representation (see \cite{ADH,Auriault,panf-bourg,ptashnyk} and references therein) that gives the average behaviour of the medium
   submitted to a macroscopic boundary condition is not working anymore.
     Specifically, now the transport becomes structured (here: {\em micro-macro\footnote{``Micro" refers here to a continuum description of a porous
      subdomain at a separated (lower) spatial scale compared to the "macro" one.}}) \cite{BLM,HJM}.

The geometry we have in mind includes  space-dependent
perforations\footnote{By ``space-dependent perforations", we mean
that at each spatial position $x$, our  model will  allow us to zoom
in  a $x$-dependent pore space, or subject to a more general
interpretation,  a  $x$-dependent porous subdomain, called here
perforation. } arranged in a locally-periodic fashion. We refer the
reader to section \ref{micromodel} (in particular to Fig.
\ref{fig1}), where we explain our concept of local periodicity. Our
approach is based on the one developed in \cite{noorden4,noorden3}
and is conceptually related to, e.g., \cite{Chechkin,Tasnim}. When
periodicity is lacking, the typical strategy would be to tackle the
matter from the percolation theory perspective (see e.g. chapter 2
in \cite{hornung} and references cited therein\footnote{Fig. 2.3 (a)
from \cite{hornung}, p. 39
   illustrates a computer simulation of the consolidation of spherical grains showing regions with high and low porosities corresponding
   to high and low diffusivity areas.}) or to reformulate the oscillating problem in terms
    of stochastic homogenization
   (see e.g. \cite{BMP}). In this paper, we stay within a
   deterministic framework by deviating in a controlled manner (made precise in section \ref{micromodel}) from the purely periodic homogenization.

We show our working methodology for a prototypical diffusion system of minimal size.
  To keep presentation simple, our scenario does not include chemistry.
      With minimal effort, both our asymptotic technique and analysis can be extended to account for volume and surface reaction production terms
       and other linear micro-macro transmission conditions.  We only emphasize the fact that if chemical reactions take place,
       then most likely that they will be hosted by the micro-structures of the low-diffusivity regions.
We discuss the microscale model for the particular case
in which the heterogenous medium is only composed
  of zones with circular areas of low diffusivity of $x$-varying sizes. This assumption on the geometry should not be seen
   as a restriction. We only use it for ease of presentations and it does not play a role in our formal and
   analytical results. Our asymptotic strategy is based on a suitable expansion
    (remotely resembling the boundary unfolding operator \cite{Zaki}) of the boundary of the perforations in terms
     of level-set functions.
     In particular, we can treat in a quite similar way situations when free-interfaces travel the
     microstructure; we refer the reader to \cite{noorden4} for a dissolution precipitation free-boundary problem and \cite{CRAS} for a fast-reaction slow-diffusion
     scenario where we addressed the matter.

The results or our paper are twofold:
\begin{itemize}
\item[(i)] We develop a strategy to deal (formally) with the asymptotics $\epsilon\to 0$ for a locally periodic medium (where $\epsilon>0$ is the microstructure width)
and derive a macroscopic equation and {\em $x$-dependent} effective
transport coefficients (porosity, permeability, tortuosity)
 for the species undergoing fast transport (i.e. that one living in high diffusivity areas), while we preserve the precise geometry
  of the microstructure and corresponding balance equation. The result of this homogenization procedure is a distributed-microstructure model in the terminology of R. E. Showalter, which we refer here as {\em two-scale model}.
\item[(ii)] We analyze the solvability of the resulting two-scale model. Solutions of the two-scale model are elements of $x$-dependent Bochner spaces.
Our approach benefits from previous work on two-scale models by, e.g.,
Showalter and Walkington \cite{Show_walk}, Eck \cite{Eck_habil}, and
Meier and B\"ohm \cite{sebam_PhD,sebam_equadiff}. A special
feature of our analysis is that most of the basic estimates
(positivity, $L^\infty$-bounds, uniqueness, energy inequality)
   are obtained in the $x$-dependent Bochner spaces. Our existence proof
   is constructed using a Schauder fixed-point argument and
   is an alternative to \cite{Show_walk}, where the situation is
   formulated as a Cauchy problem in Hilbert spaces and then resolved
   by holomorphic semigroups, or to \cite{sebam_PhD}, where a Banach-fixed point argument for the
    problem stated in transformed domains (i.e. $x$-independent) is employed.
\end{itemize}
Note that (i) and (ii) are preliminary results preparing the
framework for rigorously proving a convergence rate for the
asymptotics $\epsilon\to
  0$;
  we will address this convergence issue elsewhere.

The paper is organized in the following fashion: Section
\ref{micromodel} contains the description of the model equations at
the micro scale together with the precise geometry of our
$x$-dependent microstructure for the particular case of circular
perforations.  The homogenization procedure is detailed in section
\ref{homogenization}. The main result of this part of the  paper is
the two-scale model equations  as well as a couple of effective
coefficients reported in section \ref{upscaled}. The second part of
the paper focusses on the analysis of the two-scale model; see
section \ref{analysis}. The main result, i.e. Theorem
\ref{main_result}, ensures the global-in-time existence of weak
solutions to our two-scale model and appears at the end of section
\ref{existence}. A brief discussion section concludes the paper.

\section{Model equations}\label{micromodel}

\begin{figure}
\begin{center}
\includegraphics[width=8cm]{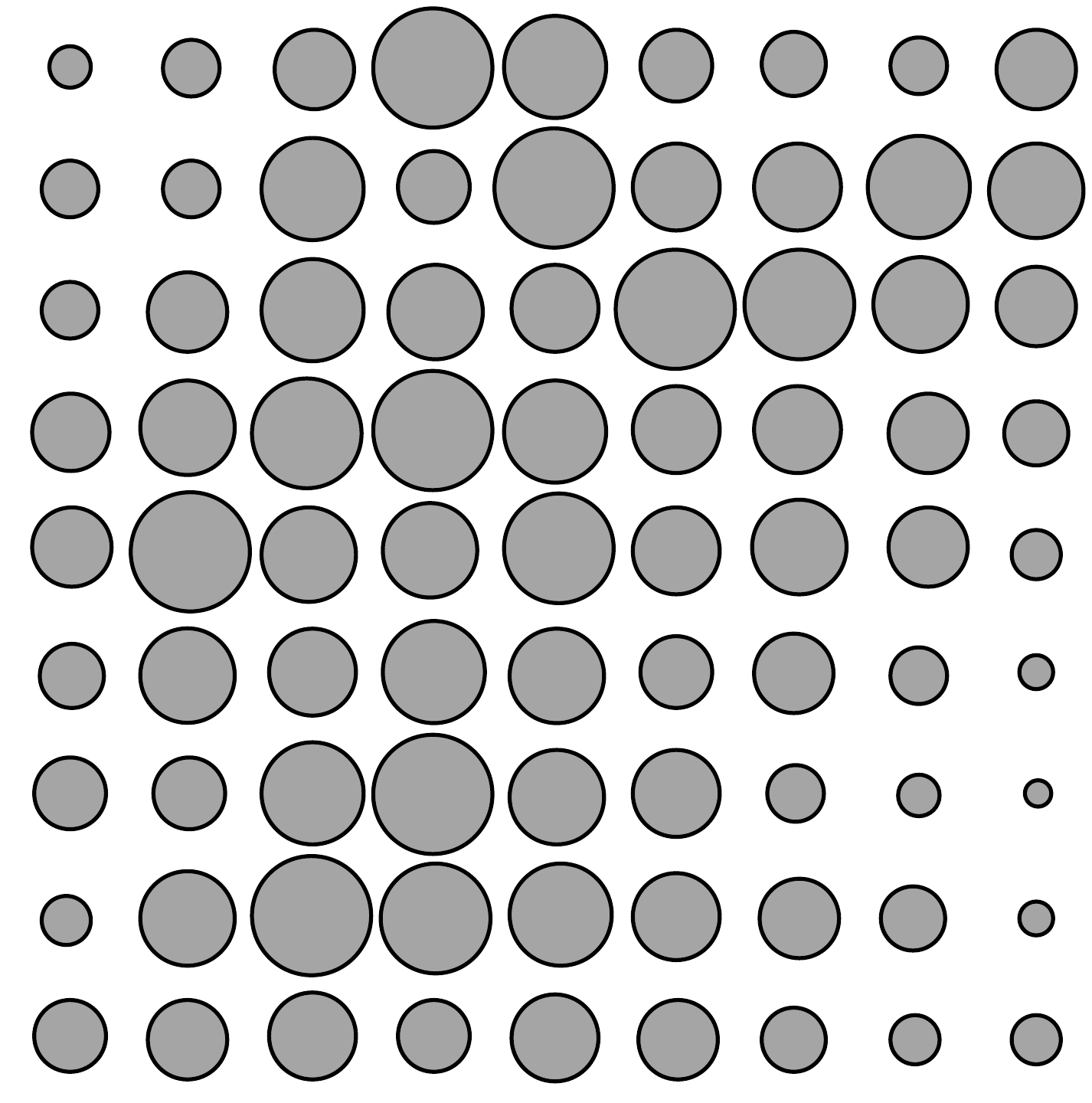}
\end{center}
\caption{Schematic representation of a locally-periodic heterogeneous medium.
The centers of the gray circles are on a grid with width $\epsilon$.
These circles represent the areas of low diffusivity and their radii may vary.
\label{fig1}}
\end{figure}
We consider a heterogenous medium consisting of areas of
high and low diffusivity.
The medium is in the present paper represented by a two dimensional domain.
We denote the two dimensional bounded domain by $\Omega\subset \mathbb{R}^2$, with boundary $\Gamma$, and for ease of presentation we suppose in this section that the areas of the medium with low
diffusivity are circles. We do not use this restriction in later sections;
there the areas with low diffusivity can have different shapes, as long as neighboring areas do not touch each other.

Let the centers of the circles $B_{ij}$ with low diffusivity, with radius $R_{ij}<\epsilon/2$, be located in a equidistant grid with nodes at $(\epsilon i, \epsilon j)$, where
$\epsilon$ is a small dimensionless length scale. We assume that there is given
a function $r(x):\Omega \rightarrow [0,1/2)$ such that
the radii $R_{ij}$ of the circles $B_{ij}$ are given
by $R_{ij}=\epsilon r(x_{ij})$, where $x_{ij}=(\epsilon i, \epsilon j)$.
We define the area of low diffusivity $\Omega^\epsilon_l$, which is the collection of the circles of low diffusivity, as $\Omega^\epsilon_l:=\cup B_{ij}$ and we define the area of high diffusivity $\Omega^\epsilon_h$, which is the complement of $\Omega^\epsilon_l$ in $\Omega$, as
$\Omega^\epsilon_h:=\Omega\backslash \Omega^\epsilon_l$.
The boundary between high and low diffusivity areas is denoted by $\Gamma^\epsilon$, which is given by $\Gamma^\epsilon:=\partial \Omega^\epsilon_l$.
It is important to note that we assume that the circles of low diffusivity
do not touch each other, so that
$\Gamma_{ij}\cap\Gamma_{kl}=\emptyset$ if $i\neq k$ or $j\neq l$, where
$\Gamma_{ij}:=\partial B_{ij}$, and we also assume that the area of low permeability does not intersect the outer boundary of the domain $\Omega$, so that
$\Gamma\cap\Gamma_{ij}=\emptyset$ for all $i,j$.

We denote the tracer concentration in the high diffusivity area by
$u^\epsilon$, the concentration in the low diffusivity area by $v^\epsilon$, the
velocity of the fluid phase by $q^\epsilon$ and the pressure
by $p^\epsilon$. All these unknowns are dimensionless.
In the high diffusivity area we assume for the fluid flow a Darcy-like law and incompressibility, while we neglect fluid flow in the low diffusivity area. The diffusion coefficient in the low diffusivity area is assumed to be of the order of $O(\epsilon^2)$, while all the remaining coefficients are of the order of $O(1)$ in $\epsilon$. We assume continuity of concentration and of fluxes across
the boundary between the high and low diffusivity areas.

The model is now given by
\begin{eqnarray}
&&\begin{cases}
u^{\epsilon}_{t}=\nabla\cdot(D_h \nabla u^{\epsilon}-q^{\epsilon}u^{\epsilon}) &\\
q^{\epsilon} = -\kappa \nabla p^{\epsilon} &\\
\nabla\cdot q^{\epsilon}=0 &
\end{cases}
\,\,\,\,\,\,\mbox{   in   } \Omega_h^{\epsilon},\label{dimlesseq1}
\\
&&
\begin{cases}
v^{\epsilon}_{t}=\epsilon^2\nabla\cdot(D_l \nabla v^{\epsilon}) &
\,\,\,\,\,\,\mbox{   in   } \Omega_l^{\epsilon},
\end{cases}\label{dimlesseq1b}
\\
&&\begin{cases}
\nu^\epsilon\cdot(D_h \nabla u^{\epsilon})=
\epsilon^2\nu^\epsilon\cdot(D_l \nabla v^{\epsilon}) & \\
u^\epsilon=v^\epsilon & \\
q^{\epsilon}=0 &
\end{cases}
\,\,\,\,\,\, \mbox{   on   } \Gamma^{\epsilon}, \label{dimlesseq2}\\
&&\begin{cases}
u^\epsilon(x,t)=u_b(x,t) & \\
q^\epsilon(x,t)=q_b(x,t) & \\
\end{cases} \,\,\,\,\,\,\, \mbox{on}\,\,\, \Gamma,\label{bcg}\\
&&\begin{cases}
u^{\epsilon}(x,0)=u^\epsilon_{I}(x)&\mbox{   in   } \Omega_h^{\epsilon},\\
v^{\epsilon}(x,0)=v^\epsilon_{I}(x)&\mbox{   in   } \Omega_l^{\epsilon},\\
\end{cases}
\label{dimlesseq3}
\end{eqnarray}
where $D_h$ denotes the diffusion coefficient in the high diffusivity
region, $D_l$ the diffusion coefficient in the low diffusivity regions, $\kappa$
denotes the permeability in the Darcy law for the flow in the high diffusivity
region, $\nu^\epsilon$ denotes the unit normal to the boundary $\Gamma^\epsilon(t)$, where $q_b$ and $u_b$ denote the Dirichlet boundary data for
the concentration $u^\epsilon$ and Darcy velocity $q^\epsilon$ and where
$u_I^\epsilon$ and $v_I^\epsilon$ denote initial value data for the concentration
$u^\epsilon$ and $v^\epsilon$.

\section{Formal homogenization}\label{homogenization}
In addition to the macroscopic variable $x$, we introduce a periodic
unit cube $U$ with microscopic variable $y$:
\begin{align}
y=(y_1,y_2),\,\, \mbox{and}\,\,
U:=\{y\in \mathbb{R}^2\,|\, -1/2\leq y_i \leq 1/2\,\, \mbox{for}\,\, i=1,2\}).
\end{align}
For the formal homogenization we assume the following formal asymptotic expansions for $u^{\epsilon}$, $v^{\epsilon}$, $q^{\epsilon}$
and $p^\epsilon$:
\begin{eqnarray*}
u^{\epsilon}(x,t)&=&u_{0}(x,x/\epsilon,t)+\epsilon u_{1}(x,x/\epsilon,t)+\epsilon^{2}u_{2}(x,x/\epsilon,t)+...\\
v^{\epsilon}(x,t)&=&v_{0}(x,x/\epsilon,t)+\epsilon v_{1}(x,x/\epsilon,t)+\epsilon^{2}v_{2}(x,x/\epsilon,t)+...\\
q^{\epsilon}(x,t)&=&q_{0}(x,x/\epsilon,t)+\epsilon q_{1}(x,x/\epsilon,t)+\epsilon^{2}q_{2}(x,x/\epsilon,t)+...\\
p^{\epsilon}(x,t)&=&p_{0}(x,x/\epsilon,t)+\epsilon p_{1}(x,x/\epsilon,t)+\epsilon^{2}p_{2}(x,x/\epsilon,t)+...
\end{eqnarray*}
where $u_{k}(\cdot,y,\cdot)$, $v_{k}(\cdot,y,\cdot)$, $q_{k}(\cdot,y,\cdot)$ and $p_k(\cdot,y,\cdot)$
are 1-periodic in $y=\frac{x}{\epsilon}$.
The gradient of a function $f(x,\frac{x}{\epsilon})$, depending on
$x$ and $y=\frac{x}{\epsilon}$ is given by
\begin{eqnarray}
\nabla f = \nabla_x f +\frac{1}{\epsilon}\nabla_y f|_{y=\frac{x}{\epsilon}},\label{chainrule}
\end{eqnarray}
where $\nabla_x$ and $\nabla_y$ denote the gradients with respect to
the first and second variables of $f$.

\subsection{Level set formulation of the perforations boundary}
Since the location of the interfaces between the low and the high diffusivity regions also depends on $\epsilon$, we need an
$\epsilon$-dependent parametrization of these interfaces.
A convenient way to parameterize the interfaces
is to use a level set function, which we denote by
$S^\epsilon(x)$:
\begin{align*}
x\in \Gamma^\epsilon \Leftrightarrow S^\epsilon(x)=0.
\end{align*}
Since we allow the size and shape of the perforations to vary with the
macroscopic variable $x$, we might use the following characterization
of $S^\epsilon$:
\begin{align}
S^\epsilon(x)=S(x,x/\epsilon)
\end{align}
where $S:\Omega\times U \rightarrow \mathbb{R}$ is 1-periodic in its second
variable, and is independent of $\epsilon$.
In this section we show, using the example of a grid of circles with varying sizes, that this characterization of $S^\epsilon$ is not sufficient to characterize all locally-periodic sequences of perforation geometries.
In fact, we need to expand $S^\epsilon$ as
\begin{align}
S^\epsilon(x)=S_0(x,x/\epsilon)+\epsilon S_1(x,x/\epsilon)+\epsilon^2S_2(x,x/\epsilon)+...
\end{align}
where the $S_i:\Omega\times U \rightarrow \mathbb{R}$ are 1-periodic in their second variable, for $i=0,1,2,...$ and are independent of $\epsilon$.

In order to find an explicit expression for $S^\epsilon(x)$ in this particular
case, i.e.\ the case of circular domains with radius $r(x)$ (see Fig.\ \ref{fig1}),
we define $P(x)$ to be the periodic extension of the function $x\rightarrow |x|$
and $Q(x)$ to be the periodic extension of the function $x\rightarrow x$,
both defined on the square $[-\frac{1}{2},\frac{1}{2}]\times[-\frac{1}{2},\frac{1}{2}]$, given by
\begin{align*}
P(x)&=P(x_1,x_2)=\sqrt{\lfloor x_1+1/2\rfloor^2+ \lfloor x_2+1/2\rfloor^2},\\
Q(x)&=Q(x_1,x_2)=(\lfloor x_1+1/2\rfloor,\lfloor x_2+1/2\rfloor),
\end{align*}
where $\lfloor a \rfloor := \max\{n\in \mathbb{Z} \,|\, n\leq a \}$
denotes the floor of $a$ (rounding down).
We can write $S^\epsilon(x)$ as follows:
\begin{eqnarray}
S^\epsilon(x)=r(x-\epsilon Q(x/\epsilon))-P(x/\epsilon).\label{TB}
\end{eqnarray}
Interestingly, the expression (\ref{TB}) plays the same role as the boundary unfolding operator (cf., for instance, \cite{Zaki} Definition 5.1).
Note that $S^\epsilon$ is not a continuous function, it jumps when $x_1$ or $x_2$
cross a multiple of $\epsilon$. Whenever we assume that $r(x,t)<1/2$, this is not
a problem, since in this case $S^\epsilon$ is continuous and smooth in a neighborhood of its zero level set, which is what we are interested in.

To check that the zero level set of $S^\epsilon$ consists indeed
of circles around $x_{ij}$ with radius $\epsilon r(x_{ij})$, we
consider a curve, which without loss of generality can be parametrized in the square with sides $\epsilon$ around $x_{ij}$ by $x_{ij}+\gamma(s)$.
For this curve to be a zero level set, it should hold that
\begin{eqnarray*}
r(x_{ij}+\gamma(s)-\epsilon Q(\epsilon^{-1}(x_{ij}+\gamma(s))))=P(\epsilon^{-1}(x_{ij}+\gamma(s))).
\end{eqnarray*}
Using that $x_{ij}=(\epsilon i,\epsilon j)$, with $\epsilon i,\epsilon j\in \epsilon \mathbb{Z}\cap \Omega$, we obtain
\begin{eqnarray*}
r( (\epsilon i,\epsilon j)+\gamma(s)-\epsilon Q((i,j)+\epsilon^{-1}\gamma(s)))=P((i,j)
+\gamma(s)),
\end{eqnarray*}
and using the periodicity of $P$ and $Q$ we get
\begin{eqnarray*}
r(x_{ij})=|\gamma(s)|,
\end{eqnarray*}
which means that $\gamma(s)$ should be a circle with radius $r(x_{ij})$.

Now we can write the level set function
$S^\epsilon$ formally as the expansion
\begin{eqnarray*}
S^{\epsilon}(x)&=&S_{0}(x,x/\epsilon)+\epsilon S_{1}(x,x/\epsilon)+\epsilon^{2}S_{2}(x,x/\epsilon)+O(\epsilon^3),
\end{eqnarray*}
where $S_k(\cdot,y,\cdot)$, for $k=0,1,2,...$, are 1-periodic in $y=\frac{x}
{\epsilon}$, and are independent of $\epsilon$.
In order to find the terms in this expansion,
we assume that $r$ is sufficiently smooth and so that we can use the
Taylor series of $r$ around $x$:
\begin{eqnarray*}
r(x-\epsilon Q(x/\epsilon)=r(x)-\epsilon Q(x/\epsilon)\cdot \nabla r(x)+
\frac{\epsilon^2}{2}Q(x/\epsilon)\cdot {\cal D}^2r(x)Q(x/\epsilon)+O(\epsilon^3),
\end{eqnarray*}
where ${\cal D}^2r$ denotes the Hessian of $r$ w.r.t.\ $x$.
This suggests the following definition of the terms in the expansion of $S^\epsilon$:
\begin{align*}
S_0(x,x/\epsilon)&:=r(x)-P(x/\epsilon),\\
S_1(x,x/\epsilon)&:=-Q(x/\epsilon) \cdot \nabla r(x),\\
S_2(x,x/\epsilon)&:=\frac{1}{2}Q(x/\epsilon)\cdot {\cal D}^2r(x)Q(x/\epsilon),
\\
&\vdots
\end{align*}

\subsection{Interface conditions}
In (\ref{dimlesseq2}$_1$) we have used the superscript $\epsilon$ for the normal
vector $\nu^\epsilon$ in the interface conditions for $v^\epsilon$ and $u^\epsilon$. The reason is that the normal vector depends
on the geometry of the different regions, and this in turn depends on
$\epsilon$. In order to perform the steps of formal homogenization, we have to expand
$\nu^\epsilon$ in a power series in $\epsilon$.
This can be done in terms of the level set function $S^\epsilon$:
\begin{align}
\nu^\epsilon=\frac{\nabla S^\epsilon(x,x/\epsilon)}{|\nabla S^{\epsilon}(x,x/\epsilon)|}\,\,\, \mbox{at}\,\,\, x\in \Gamma^\epsilon.
\end{align}
First we expand $|\nabla S^{\epsilon}|$. Using the
chain rule \eqref{chainrule} (see also \cite{hornung}),
the expansion of $S^{\epsilon}$ and the Taylor series of the
square-root function, we obtain
\begin{eqnarray}
|\nabla S^{\epsilon}|&=&
\frac{1}{\epsilon}|\nabla_{y} S_0|+O(\epsilon^0). \label{nabsex}
\end{eqnarray}
In the same fashion, we get
\begin{eqnarray*}
\nu^{\epsilon}=\nu_{0}+\epsilon \nu_{1}+O(\epsilon^{2}),
\end{eqnarray*}
where
\begin{eqnarray*}
\nu_{0}:=\frac{\nabla_{y}S_0}{|\nabla_{y}S_0|}
\end{eqnarray*}
and
\begin{eqnarray*}
\nu_{1}&:=&\frac{\nabla_{x}S_0+\nabla_yS_1}{|\nabla_{y}S_0|}-
\frac{(\nabla_{x}S_0\cdot \nabla_{y}S_0+\nabla_yS_0\cdot \nabla_y S_1)}{|\nabla_{y}S_0|^{2}}\frac{\nabla_{y}S_0}{|\nabla_{y}S_0|}.
\end{eqnarray*}
If we introduce the normalized tangential vector $\tau_{0}$, with
$\tau_0\perp\nu_{0}$,
we can rewrite $\nu_1$ as
\begin{eqnarray}
\nu_{1}&=&\tau_{0}\frac{\tau_{0}\cdot(\nabla_{x}S_0+\nabla_yS_1)}{|\nabla_{y}S_0|}.\label{nu1}
\end{eqnarray}
Now we focus on the interface conditions posed at $\Gamma^\epsilon$.
In order to obtain interface conditions in the auxiliary problems, we
substitute the
expansions of $u^{\epsilon}$, $q^\epsilon$, and $\nu^\epsilon$ into \eqref{dimlesseq2}.
This is not so straight-forward as it may seem, since the interface
conditions \eqref{dimlesseq2} are enforced at the oscillating interface
$\Gamma^{\epsilon}$, i.e.
at every $x$ where $S^{\epsilon}(x)=0$.
For formulating the upscaled model it would be convenient
to have boundary conditions enforced at
\begin{align}
\Gamma_{0}(x):=\{y\,|\,S_0(x,y)=0\}.
\end{align}
To obtain them, we suppose that we can parametrize
the part of the boundary $\Gamma^{\epsilon}_{ij}$ that surrounds the sphere $B_{ij}$ with $k^{\epsilon}(s)$, so that holds
\begin{eqnarray*}
S^{\epsilon}(k^{\epsilon}(s))=0,
\end{eqnarray*}
and we assume that we can
expand $k^{\epsilon}(s)$ using the formal asymptotic
expansion
\begin{eqnarray}
k^{\epsilon}(s)=x_{ij}+\epsilon k_{0}(s)+\epsilon^{2} k_{1}(s)+O(\epsilon^{3}).
\label{kexp}
\end{eqnarray}
Using the expansion for $S^{\epsilon}$, the periodicity of $S_{i}$ in $y$,
and the Taylor series of $S_{0}$ and $S_{1}$ around $(x,k_0)$, we obtain
\begin{eqnarray*}
S_{0}(x,k_{0})+\epsilon(S_{1}(x,k_{0})+k_{0}\cdot \nabla_{x}S_{0}(x,k_{0})
+k_{1}\cdot\nabla_{y}S_{0}(x,k_{0}))+O(\epsilon^{2})=0.
\end{eqnarray*}
Collecting terms with the same order of $\epsilon$, we see that
$k_{0}(s)$ parametrizes locally the zero level set of $S_{0}$:
\begin{eqnarray*}
S_{0}(x,k_{0})=0.
\end{eqnarray*}
For $k_{1}$, we have the equation
\begin{eqnarray}
S_{1}(x,k_{0})+k_{0}\cdot \nabla_{x}S_{0}(x,k_{0})
+k_{1}\cdot\nabla_{y}S_{0}(x,k_{0})=0. \label{k1eq}
\end{eqnarray}
It suffices to seek for $k_1$ that is aligned with $\nu_0$, so that
we write
\begin{eqnarray}
k_{1}(s)=\lambda(s))\nu_{0}(s)=\lambda\frac{\nabla_{y} S_{0}}{|\nabla_{y}S_{0}|},\label{exprk}
\end{eqnarray}
where, using \eqref{k1eq}, $\lambda$ is given by
\begin{eqnarray}
\lambda:=-\frac{S_{1}}{|\nabla_{y}S_{0}|}-\frac{k_{0}\cdot
\nabla_{x}S_{0}}{|\nabla_{y}S_{0}|}. \label{explam}
\end{eqnarray}
Each of the boundary conditions in \eqref{dimlesseq2} admits the structural
form
\begin{eqnarray*}
K(x,x/\epsilon)=0\,\mbox{   for all }\,x\in\Gamma^{\epsilon},
\end{eqnarray*}
where $K$ is a suitable linear combination of $u^\epsilon$, $\nabla u^\epsilon$, $q^\epsilon$, $p^\epsilon$, $v^\epsilon$, and $\nabla v^\epsilon$.
Using \eqref{kexp} and the Taylor series of $K$ around $(x,k_{0})$, we obtain
\begin{eqnarray}
K(x,k_{0})&+&\epsilon(k_{0}\cdot\nabla_{x}K(x,k_{0})+k_{1}\cdot\nabla_{y}K(x,k_{0})) \nonumber\\
&&+\frac{\epsilon^2}{2}(k_0,k_1)\cdot ({\cal D}^2K(x,k_0))(k_0,k_1)+\epsilon^3(...)=0,\label{Ktay}
\end{eqnarray}
where ${\cal D}^2K$ denotes the Hessian of $K$ w.r.t.\ $x$ and $y$.
Substituting \eqref{exprk} into \eqref{Ktay}, we can restate \eqref{Ktay} in the
following way:
\begin{eqnarray}
K(x,y)&+&\epsilon(y\cdot\nabla_{x}K(x,y)+\lambda\nu_{0}\cdot\nabla_{y}
K(x,y))\nonumber\\
&&+\frac{\epsilon^2}{2}(y,\lambda\nu_0)\cdot ({\cal D}^2K(x,y))(y,\lambda\nu_0)+O(\epsilon^3)=0
\,\,\mbox{for all}\,y\in \Gamma_{0}(x).\label{Ktay2}
\end{eqnarray}
In order to proceed further, we make use of the following technical lemmas.
Their proofs can be found in \cite{noorden4}.
\begin{lemma}\label{lemma1}
Let $g(x,y)$ be a scalar function such that $g(x,y)=0$ for all
$y\in \Gamma_0(x)$, $x\in\Omega$ and $t\geq 0$. Then it holds that
\begin{equation*}
\nabla_xg=\frac{\nu_0\cdot\nabla_yg}{|\nabla_y S_0|}\nabla_x S_0,
\,\,\,\mbox{for}\,\,\,x\in\Omega,\,\,y\in\Gamma_0(x,t).
\end{equation*}
\end{lemma}

\begin{lemma}\label{lemma2}
Let $F(x,y)$ be a vector valued function such that $\nabla_y\cdot F(x,y)=0$
on $Y_0(x):=\{y\,|\,S_0(x,y)>0\}$
and $\nu_0\cdot F(x,y)=0$ on
$\Gamma_0(x)$ for all $x\in\Omega$. Then it holds that
\begin{equation*}
\int_{\Gamma^{0}(x)}\frac{\tau_{0}\cdot\nabla_{y}S_{1}}
{|\nabla_{y}S_{0}|}\tau_{0}\cdot F  - \frac{S_{1}}{|\nabla_{y}S_{0}|}
\nu_{0}\cdot\nabla_{y}(\nu^{0}\cdot F)\,d\sigma=0,
\,\,\,\mbox{for}\,\,\,
x\in\Omega.
\end{equation*}
\end{lemma}

\subsection{Flow equations}
Substituting the asymptotic expansions of $q^\epsilon$ and
$p^\epsilon$ into (\ref{dimlesseq1}$_{2,3}$),
we obtain
\begin{eqnarray}
&&q_{0}=-\kappa\frac{1}{\epsilon}\nabla_{y}p_{0}-\kappa\nabla_{y}p_{1}-
\kappa\nabla_{x}p_{0}+O(\epsilon), \label{eqqpex}\\
&&\frac{1}{\epsilon}\nabla_{y}\cdot q_{0}+\nabla_{x}\cdot q_{0}+\nabla_{y}
\cdot q_{1}+O(\epsilon)=0.\label{eqqex}
\end{eqnarray}
Substituting the asymptotic expansion of $q^\epsilon$ into the boundary condition
(\ref{dimlesseq2}$_3$), and using \eqref{Ktay2}, gives
\begin{eqnarray}
q_{0}+\epsilon\Big(q_{1}+(\nabla_{x}q_{0})^Ty+\lambda(\nabla_{y}q_{0})^T\nu_0\Big)+O(\epsilon^{2})=0,\,\,\,
\mbox{for all}\,\,y\in\Gamma_0(x).\,\,\,\,\,\,\,\,\label{beqqex}
\end{eqnarray}
The $\epsilon^{-1}$-term in \eqref{eqqpex} indicates that
$\nabla_{y}p_{0}=0$,
so that we conclude that $p_{0}$ is independent of $y$.
Furthermore, we obtain, after collecting $\epsilon^0$-terms from
\eqref{eqqpex} and \eqref{beqqex} and $\epsilon^{-1}$-terms from
\eqref{eqqex}, the equations for $q_{0}$ and $p_{1}$:
\begin{eqnarray}
\begin{cases}
q_{0}=-\kappa\nabla_{y}p_{1}-\kappa\nabla_{x}p_{0} & \mbox{in}\,\,\,Y_{0}(x),\\
\nabla_{y}\cdot q_{0}=0 & \mbox{in}\,\,\, Y_{0}(x),\\
q_{0}=0&\mbox{on}\,\,\,\Gamma_{0}(x),\\
\mbox{$q_0$ and $p_0$ $y$-periodic},&
\end{cases}\label{cpv}
\end{eqnarray}
where
\begin{align}
Y_0(x):=\{y\,|\,S_0(x,y)>0\}.\label{defy0}
\end{align}
These equations (together with boundary conditions on the outer boundary $\partial \Omega$) determine the averaged velocity field given by
\begin{eqnarray*}
\bar{q}(x)=\int_{Y_0(x)}q_0(x,y)\,dy.
\end{eqnarray*}
Now we compute the divergence of $\bar{q}$ (where we use the
$\epsilon^0$-terms from \eqref{eqqex})
\begin{eqnarray*}
\nabla_{x}\cdot\bar{q}&=&\nabla_{x}\cdot\int_{Y_0(x)}q_{0}\,dy=
\int_{Y_0(x)}\nabla_{x}\cdot q_{0}\,dy
-\int_{\Gamma_{0}(x)}\frac{\nabla_{x}S_{0}}
{|\nabla_{y}S_{0}|}\cdot q_{0}
\,d\sigma\\
&=&-\int_{Y(x)}\nabla_{y}\cdot q_{1}\,dy=-\int_{\Gamma_{0}(x)}\nu_{0}
\cdot q_{1}\,d\sigma\\
&=&\int_{\Gamma_{0}(x)}
-\nu_{0}\cdot((\nabla_{x}q_{0})^Ty+\lambda (\nabla_{y}
q_{0})^T\nu_0)\, d\sigma\\
&=&
-I_1-I_2,
\end{eqnarray*}
with
\begin{eqnarray*}
&&I_1:=\int_{\Gamma_{0}(x)}\nu_{0}\cdot\Big((\nabla_{x}q_{0})^Ty-
\frac{y\cdot \nabla_x S_0}{|\nabla_yS_0|}(\nabla_yq_0)^T\nu_0\Big)\,d\sigma,\\
&&I_2:=-\int_{\Gamma_{0}(x)}\nu_0\cdot\Big(\frac{S_1}{|\nabla_yS_0|}
(\nabla_yq_0)^T\nu_0\Big)\,d\sigma.
\end{eqnarray*}
We apply Lemma \ref{lemma1} with $g=\nu_0\cdot q_0$, and obtain
\begin{equation*}
\nabla_x(\nu_0\cdot q_0)=\frac{\nu_0\cdot\nabla_{y}(\nu_0\cdot q_0)}
{|\nabla_{y}S_{0}|}\nabla_{x}S_{0},\,\,\,\mbox{on}\,\,\,\Gamma_0(x,t).
\end{equation*}
Since $q_0=0$ on $\Gamma_0(x)$ it follows that
$(\nabla_xq_0)^T\nu_0=\frac{\nu_0\cdot(\nabla_yq_0)^T\nu_0}
{|\nabla_yS_0|}\nabla_xS_0$, so that $I_1=0$.
Next we apply Lemma \ref{lemma2} with $F=q_0$, and get consequently
\begin{eqnarray*}
\int_{\Gamma^{0}(x)}\frac{\tau_{0}\cdot\nabla_{y}S_{1}}
{|\nabla_{y}S_{0}|}\tau_{0}\cdot q_0  - \frac{S_{1}}{|\nabla_{y}S_{0}|}
\nu_{0}\cdot\nabla_{y}(\nu^{0}\cdot q_0)\,d\sigma=0.
\end{eqnarray*}
Again using $q_0=0$ on $\Gamma_0(x)$, it follows that
$I_2=0$, so that we have
\begin{align}
\nabla_x\cdot\bar{q}=0.\label{divq}
\end{align}

\subsection{Diffusion equation in the low diffusivity areas}
Substituting the asymptotic expansion of $v^\epsilon$ into \eqref{dimlesseq1b}, we obtain
\begin{eqnarray}
\partial_tv_0=D_l\nabla_y v_0+O(\epsilon). \label{eqv0}
\end{eqnarray}
Similarly expanding the boundary condition (\ref{dimlesseq2}$_2$), we get
\begin{eqnarray*}
0=u_0-v_0+O(\epsilon)\,\mbox{   on  }\,\,\Gamma^\epsilon,
\end{eqnarray*}
which, after substitution into \eqref{Ktay2}, becomes
\begin{eqnarray*}
0=u_0-v_0+O(\epsilon) \,\mbox{   on  }\,\,\Gamma_0(x).
\end{eqnarray*}
Collecting the lowest order terms, and using that $u_0$ does not depend on $y$,
we obtain the boundary condition
\begin{eqnarray}
v_0(x,y,t)=u_0(x,t) \,\mbox{   for all}\,\,y\in \Gamma_0(x),\,x\in\Omega.
\label{uisv}
\end{eqnarray}

\subsection{Convection-diffusion equation in the high diffusivity area}
Substituting the asymptotic expansion of $u^{\epsilon}$ into (\ref{dimlesseq1}$_{1}$), we obtain
\begin{eqnarray}
\partial_{t}u_{0}&=&\frac{1}{\epsilon^{2}}D_h\Delta_{y}u_{0}+
\frac{1}{\epsilon}(\nabla_{y}\cdot
F_h+\nabla_{x}\cdot(D_h\nabla_{y}u_{0})) \nonumber \\
&&+\nabla_{y}\cdot(D_h(\nabla_{y}u_{2}+\nabla_{x}u_{1})-q_{1}u_{0}-q_{0}u_{1})
+\nabla_{x}\cdot
F_h\,\,\,\label{equex}\\
&&+O(\epsilon), \nonumber
\end{eqnarray}
where
\begin{align}
F_h:=D_h(\nabla_{x}u_{0}+\nabla_{y}u_{1})-q_{0}u_{0}.
\end{align}
Using the expansions for $u^{\epsilon}$, $v^{\epsilon}$ and $\nu^{\epsilon}$,
we first expand (\ref{dimlesseq2}$_1$):
\begin{eqnarray*}
0&=&\nu^{\epsilon}\cdot (D_h\nabla u^{\epsilon})-\epsilon^2\nu^\epsilon \cdot
(D_l\nabla v_\epsilon)\\
&=&\frac{1}{\epsilon}\nu_{0}\cdot (D_h\nabla_{y}u_{0})+\nu_0\cdot (D_h(\nabla_{x}u_{0}+\nabla_{y}u_{1}))+\nu_{1}\cdot
(D_h\nabla_{y}u_{0})\\
&+&\epsilon\Big(\nu_{0}\cdot(D_h(\nabla_{x}u_{1}+\nabla_{y}u_{2}))
+\nu_{1}\cdot (D_h(\nabla_{x}u_{0}+\nabla_{y}u_{1}))
+\nu_{2}\cdot (D_h\nabla_{y}u_{0})-\nu_0 \cdot (D_l\nabla_yv_0)
\Big)\\
&+&O(\epsilon^{2}), \,\,\mbox{for all}\,\, x\in\Gamma^\epsilon\, \mbox{and}\,y=\frac{x}{\epsilon}.
\end{eqnarray*}
Next we substitute this expansion into \eqref{Ktay2}, and thus obtain
\begin{eqnarray}
0&=&\frac{1}{\epsilon}\nu_{0}\cdot(D_h\nabla_{y}u_{0})\nonumber\\
&+&\nu_{0}\cdot (D_h(\nabla_{x}u_{0}+\nabla_{y}u_{1}))
+\nu_{1}\cdot(D_h\nabla_{y}u_{0})
+y\cdot\nabla_{x}(\nu_{0}\cdot(D_h\nabla_{y}u_{0}))+\lambda\nu_{0}\cdot\nabla_{y}(\nu_{0}\cdot(D_h\nabla_{y}u_{0}))\nonumber\\
&+&\epsilon\Big(\nu_{0}\cdot(D_h(\nabla_{x}u_{1}+\nabla_{y}u_{2}))
+\nu_{1}\cdot D_h(\nabla_{x}u_{0}+\nabla_{y}u_{1})
+\nu_{2}\cdot(D_h\nabla_{y}u_{0})\nonumber\\
&&-\nu_0\cdot(D_l\nabla_yv_0)
+y\cdot\nabla_{x}(\nu_{0}\cdot (D_h(\nabla_{x}u_{0}+\nabla_{y}u_{1}))
+\nu_{1}\cdot(D_h\nabla_{y}u_{0}))\nonumber\\
&&+\lambda\nu_{0}\cdot\nabla_{y}(\nu_{0}\cdot (D_h(\nabla_{x}u_{0}+\nabla_{y}u_{1}))
+\nu_{1}\cdot(D_h\nabla_{y}u_{0}))\nonumber\\
&&+\frac{1}{2}(y,\lambda\nu_0)\cdot ({\cal D}^2(\nu_0\cdot(D_h\nabla_yu_0)))(y,\lambda\nu_0)
\Big)\nonumber\\
&+&O(\epsilon^{2}), \,\,\,\mbox{for}\,y\in \Gamma_{0}(x).
\label{bequex}
\end{eqnarray}
Now we collect the $\epsilon^{-2}$-term from \eqref{equex}
and the $\epsilon^{-1}$-term from \eqref{bequex}. Hence we obtain for $u_{0}$ the equations
\begin{eqnarray}
\begin{cases}
\Delta_{y}u_{0}=0&\mbox{in $Y_0(x)$},\\
\nu_0\cdot \nabla_{y}u_{0}=0&\mbox{on $\Gamma_0(x)$},\\
\mbox{$u_0$ $y$-periodic},&
\end{cases}
\end{eqnarray}
where $Y_0(x)$ is given by \eqref{defy0}.
This means that $u_{0}$ is determined up to a constant and does not
depend on $y$, so that $\nabla_yu_0=0$.
Collecting the $\epsilon^{-1}$ terms from \eqref{equex}, the
$\epsilon^0$-terms from \eqref{bequex}, and using that
$\nabla_yu_0=0$, we get for $u_1$ the equations
\begin{eqnarray}
\begin{cases}
\nabla_y\cdot(D_h\nabla_{y}u_{1}-q_0u_0)=0 & \mbox{in $Y_0(x)$},\\
\nu_0\cdot(D_h(\nabla_{x}u_{0}+\nabla_{y}u_{1}))=0&\mbox{on $\Gamma_0(x)$},\\
\mbox{$u_1$ $y$-periodic}.&
\end{cases}\label{u1eq}
\end{eqnarray}
Collecting the $\epsilon^0$-terms from \eqref{equex}
and the $\epsilon^1$-terms from \eqref{bequex}, we obtain
\begin{eqnarray}
\begin{cases}
\partial_{t}u_{0}=\nabla_{y}\cdot(D_h(\nabla_{y}u_{2}+
\nabla_{x}u_{1})-q_{1}u_{0}-q_{0}u_{1})+
\nabla_{x}\cdot F_h & \mbox{in $Y_{0}(x)$},\\
\nu_{0}\cdot(D_h(\nabla_{x}u_{1}+\nabla_{y}u_{2}))=-\nu_{1}\cdot (D_h(\nabla_{x}u_{0}+\nabla_{y}u_{1}))&\\
\hspace{1cm}+\nu_0\cdot(D_l\nabla_yv_0)-y\cdot\nabla_{x}(\nu_{0}\cdot (D_h(\nabla_{x}u_{0}+\nabla_{y}u_{1})))&\\
\hspace{1cm}-\lambda \nu_{0}\cdot\nabla_{y}(\nu_{0}\cdot (D_h(\nabla_{x}u_{0}+\nabla_{y}u_{1}))) &
\mbox{on $\Gamma_{0}(x)$},\\
\mbox{$u_2$ $y$-periodic}.&
\end{cases}
\label{uzero}
\end{eqnarray}
Integrating (\ref{uzero}$_1$) over $Y_0(x)$ and using
the boundary conditions (\ref{cpv}$_3$) and (\ref{uzero}$_2$) yields
\begin{eqnarray*}
|Y_{0}(x)|\partial_{t}u_{0}&=&\int_{Y_{0}(x)}\nabla_{y}\cdot(D_h(\nabla_{x}u_{1}+\nabla_{y}u_{2})-q_{1}u_{0}-q_{0}u_{1})\,dy
+\int_{Y_{0}(x)} \nabla_{x}\cdot F_h
\,dy\\
&=&\int_{\Gamma_{0}(x)}-\nu_{1}\cdot F_h
+\nu_0\cdot(D_l\nabla_yv_0)
-y\cdot\nabla_{x}(\nu_{0}\cdot F_h)
-\lambda \nu_0\cdot\nabla_{y}(\nu_{0}\cdot F_h)
\,d\sigma\\
&&+\nabla_{x}\cdot\int_{Y_{0}(x)} F_h
\,dy+\int_{\Gamma_{0}(x)}\frac{\nabla_{x}S_{0}}
{|\nabla_{y}S_{0}|}\cdot F_h
\,d\sigma.
\end{eqnarray*}
Using \eqref{nu1}, \eqref{explam}, and the boundary conditions (\ref{cpv}$_3$) and (\ref{u1eq}$_2$), this can be rewritten as
\begin{eqnarray*}
|Y_0(x)|\partial_tu_{0}
&=&\nabla_{x}\cdot\int_{Y_{0}(x)}(D_h(\nabla_{y}u_{1}+\nabla_{x}u_{0})-q_{0}u_{0})\,dy\\
&&+\int_{\Gamma_0(x)}\nu_0\cdot (D_l\nabla_yv_0)\,dy
-I_{1}-I_2,
\end{eqnarray*}
where
\begin{eqnarray*}
&&I_{1}:=\int_{\Gamma_{0}(x)}y\cdot\nabla_{x}g
-\frac{y\cdot\nabla_{x}S_{0}}{|\nabla_{y}S_{0}|}
\nu_{0}\cdot\nabla_{y}g\,d\sigma,\\
&&I_{2}:=\int_{\Gamma_{0}(x)}\frac{\tau_{0}\cdot\nabla_{y}S_{1}}
{|\nabla_{y}S_{0}|}\tau_{0}\cdot F_h  - \frac{S_{1}}{|\nabla_{y}S_{0}|}
\nu_{0}\cdot\nabla_{y}(\nu^{0}\cdot F_h)\,d\sigma,
\end{eqnarray*}
with $g:=\nu_{0}\cdot F_h$,
The boundary conditions (\ref{cpv}$_3$)
and (\ref{u1eq}$_2$) give us $g(x,y,t)=0$ for $y\in\Gamma_0(x,t)$.
Now invoking Lemma \ref{lemma1} leads to
$\nabla_xg=\frac{\nu_0\cdot\nabla_yg}{|\nabla_yS_0|}\nabla_xS_0$.
So $I_1=0$.
For the integral $I_2$ we invoke Lemma \ref{lemma2} to obtain $I_2=0$.
As a last step, we use the divergence theorem and interface condition \eqref{uisv} to obtain
\begin{eqnarray}
\partial_t\left(|Y_0(x)| u_{0}+\int_{Y^C_0(x)} v_0\,dy\right)
&=&\nabla_{x}\cdot\int_{Y_{0}(x)}(D_h(\nabla_{y}u_{1}+\nabla_{x}u_{0})-q_{0}u_{0})\,dy, \label{equ0}
\end{eqnarray}
where $Y^C_{0}(x)$ is the complement of $Y_{0}(x)$ in $U$ given by
$Y^C_{0}(x):=U\backslash Y_{0}(x)=\{S_0(x)<0\}$.

\begin{remark}
Note that in this section we have not used any assumptions of the shape of
the perforations. They may have any shape as long as their limiting
shape is described by the level set function $S_0$.
\end{remark}

\section{Upscaled equations}\label{upscaled}
The equations for lowest order terms of $q^\epsilon$ and $p^\epsilon$,
\eqref{cpv} and \eqref{divq}, $v^\epsilon$, \eqref{eqv0}, $u^\epsilon$, \eqref{equ0}, and the coupling conditions \eqref{uisv} together constitute the upscaled
model. In this section we collect these equations for the case discussed in Section \ref{micromodel}, i.e.\ for circular perforations.
For this purpose we return to a formulation in terms
of $r(x,t)$, where we use
\begin{align*}
&\Gamma_0(x)=\{y\in U \,|\, |y|=r(x)\},\\
&Y_0(x)=\{y \in U\, |\, |y|>r(x)\},\\
&Y^C_{0}(x)=\{y \in U\, |\, |y|<r(x)\}.
\end{align*}
We write the solutions of equations \eqref{u1eq} and \eqref{cpv} in terms
of the solutions of the following two cell problems (see, e.g.\ \cite{hornung})
\begin{eqnarray}
\begin{cases}
\Delta_{y}v_{j}(x,y)=0 & \mbox{for all}\,\,\, x\in\Omega,\, y\in U,\,|y|>r(x),\\
\nu_{0}\cdot\nabla_{y}v_{j}(x,y)=-\nu_0\cdot e_{j}& \mbox{for all}\,\,\,x\in\Omega,\, |y|=r(x),\\
\mbox{$v_j(x,y)$ $y$-periodic},&
\end{cases} \label{cellpu}
\end{eqnarray}
and
\begin{eqnarray}
\begin{cases}
w_{j}(x,y)=\nabla_{y}\pi_{j}(x,y)+e_{j} & \mbox{for all}\,\,\,x\in\Omega,\, y\in U,\, |y|>r(x),\\
\nabla_{y}\cdot w_{j}(x,y)=0 & \mbox{for all}\,\,\, x\in\Omega,\,y\in U,\, |y|>r(x),\\
w_{j}=0&\mbox{for all}\,\,\,x\in\Omega,\, |y|=r(x),\\
\mbox{$w_j(x,y)$ and $\pi_j(x,y)$ $y$-periodic},&
\end{cases}\label{cellpq}
\end{eqnarray}
for $j=1,2$.
The use of these cell problems allows us to write the results of the
formal homogenization procedure in the form of the following
distributed-microstructure model
\begin{eqnarray}\label{upsceq}
&&\begin{cases}
\partial_{t}v_{0}(x,y,t)=D_l \Delta_y v_0(x,y,t)& \mbox{for}\,\,\,|y|<r(x),\,\,x\in\Omega,\\
\partial_{t}\left(\theta(x)u_{0}+\int_{|y|<r(x)} v_0\,dy\right)= \\
\hspace{4cm}\nabla_{x}\cdot
(D_h{\cal A}(x)\nabla_{x}u_{0}-\bar{q}u_0)&\mbox{for}\,\,\,x\in
\Omega,\\
\bar{q}=-\kappa{\cal K}(x)\nabla_{x}p_{0}&\mbox{for}\,\,\,x\in
\Omega,\\
\nabla_{x}\cdot\bar{q}=0
&\mbox{for}\,\,\,x\in
\Omega,
\end{cases} \label{upsc1}\\
&&\begin{cases}
v_0(x,y,t)=u_0(x,t)&\mbox{for}\,\,\,|y|=r(x),\\
u_0(x,t)=u_b(x,t)&\mbox{for}\,\,\,x\in\Gamma,\\
\bar{q}(x,t)=q_b(x,t)&\mbox{for}\,\,\,x\in\Gamma,
\end{cases}\label{upsc2}\\
&&\begin{cases}
u_0(x,0)=u_I(x)&\mbox{for}\,\,\,x\in
\Omega,\\
v_0(x,y,0)=v_I(x,y)& \mbox{for}\,\,\,|y|<r(x),\,\,x\in\Omega.
\end{cases}\label{upsc3}
\end{eqnarray}
where the porosity $\theta(x)$ of the medium is given by
\begin{align*}
\theta(x):=1-\pi r^2(x),
\end{align*}
while the effective diffusivity ${\cal A}(x):=(a_{ij}(x))_{i,j}$
and the effective permeability ${\cal K}(x):=(k_{ij}(x))_{i,j}$ are defined by
\begin{eqnarray*}
a_{ij}(x):=\int_{\{y \in U\, |\, |y|>r(x)\}}\delta_{ij}+\partial_{y_{i}}v_{j}(x,y,t)\, dy,
\end{eqnarray*}
and
\begin{eqnarray*}
k_{ij}(x):=\int_{\{y \in U\, |\, |y|>r(x)\}}w_{ji}(x,y,t)\,dy.
\end{eqnarray*}


\section{Analysis of upscaled equations}\label{analysis}

In this section we investigate the solvability of the upscaled
equations \eqref{upsc1}-\eqref{upsc3}. Note that the equations
(\ref{upsc1}$_{3,4}$) for $\bar{q}$ and $p_0$, together with the
boundary condition (\ref{upsc2}$_3$) are decoupled from the other
equations. We may assume that we can solve these equations for
$\bar{q}$ and $p_0$ such that $q\in L^\infty(\Omega;\mathbb{R}^2)$
(see Assumption 2 below). Standard arguments form the theory of
partial differential equations justify this assumption if the data
$q_b$ and $r$ are suitable, see \cite{HJ91} for a closely related
scenario. With this assumption the equations
\eqref{upsc1}-\eqref{upsc3} reduce to the following problem
\begin{align*}
(P)
\begin{cases}
\theta(x)\partial_tu-\nabla_x\cdot(D(x)\nabla_x u -qu)=-\int_{\partial B(x)}
\nu_y\cdot(D_l\nabla_y v)\,d\sigma & \mbox{in}\,\, \Omega,\\
\partial_tv-D_l \Delta_y v=0 & \mbox{in}\,\, B(x)
,\\
u(x,t)=v(x,y,t)& \mbox{at}\,\, (x,y)\in \Omega\times \partial B(x),\\
u(x,t)=u_b(x,t) & \mbox{at}\,\, x\in \partial \Omega,\\
u(x,0)=u_I(x) & \mbox{in}\,\,\overline{\Omega},\\
v(x,y,0)=v_I(x,y) & \mbox{at}\,\,(x,y)\in \overline{\Omega}\times\overline{B(x)},
\end{cases}
\end{align*}
where $B(x):=Y_0(x)$, where $Y_0$ is defined in \eqref{defy0}.
Notice that in this section we again do not restrict ourselves to circular
perforations. The perforations may have any shape as long as they are described
by the level set $S_0$.
In the following sections we discuss the existence and uniqueness of weak solutions to problem $(P)$.

\subsection{Functional setting and weak formulation}
For notational convenience we define the following spaces:
\begin{align}
&V_1:=H_0^1(\Omega),\\
&V_2:=L^2(\Omega;H^2(B(x))),\\
&H_1:=L^2_\theta(\Omega),\\
&H_2:=L^2(\Omega;L^2(B(x))).
\end{align}
The spaces $H_2$ and $V_2$ make sense if, for instance, we assume (like in \cite{sebam_equadiff}):

\begin{assumption}
The function
$S_0:\Omega\times U \to \mathbb{R}$, which defines $B(x):=Y_0(x)$ in \eqref{defy0}, and which also defines the
1-dimensional boundary $\Omega\times \partial B(x)$ of  $\Omega\times B(x)$ as
$$(x,y)\in \Omega\times \partial B(x) \mbox{ if and only if } S_0(x,y)=0, $$
is an element of $C^2(\overline{\Omega\times U})$.
Assume additionally that the Clarke gradient $\partial_yS_0(x,y)$ is regular for all choices of $(x,y)\in \overline{\Omega\times U}$.
\end{assumption}

Following the lines of \cite{sebam_equadiff} and \cite{Show_walk}, Assumption 1  implies in particular that the measures $|\partial B(x)|$ and $|B(x)|$ are bounded away from zero (uniformly in $x$). Consequently, the following direct Hilbert integrals (cf. \cite{Dix} (part II, chapter 2), e.g.)
\begin{eqnarray}
L^2(\Omega;H^1(B(x)))&:=&\{u\in L^2(\Omega;L^2(B(x))): \nabla_y u\in L^2(\Omega;L^2(B(x)))\}\nonumber\\
L^2(\Omega;H^1(\partial B(x)))&:=& \{u:\Omega\times \partial B(x)\to \mathbb{R} \mbox{ measurable  such that }  \int_\Omega ||u(x)||^2_{L^2(\partial B(x))}<\infty\} \nonumber
\end{eqnarray}
are well-defined separable Hilbert spaces and, additionally, the {\em distributed trace}
$$\gamma: L^2(\Omega;H^1(B(x)))\to L^2(\Omega, L^2(\partial B(x)))$$ given by
\begin{equation}\label{g}
\gamma u(x,s):=(\gamma_x U(x))(s), \ x\in \Omega, s\in \partial B(x),u\in L^2(\Omega;H^1(B(x)))
\end{equation}
is a bounded linear operator.  For each fixed $x\in \Omega$, the map $\gamma_x$, which is arising in (\ref{g}), is the standard trace operator from $H^1(B(x)$ to $L^2(\partial B(x))$. We refer the reader to \cite{sebam_PhD} for more details on the construction of these spaces and to \cite{sf} for the definitions of their duals as well as for  a less regular condition (compared to Assumption 1) allowing to define these spaces in the context of a certain class of anisotropic Sobolev spaces.

Furthermore we assume

\begin{assumption}
\begin{align*}
\begin{cases}
\theta,\, D \in L^\infty_+(\Omega),&\\
q\in L^\infty(\Omega;\mathbb{R}^d)\,\,\mbox{with}\,\, \nabla\cdot q =0,&\\
u_b\in L_+^\infty(\Omega\times S)\cap H^1(S;L^2(\Omega)),&\\
\partial_tu_b\leq 0\,\, \mbox{a.e.}\,\,(x,t)\in\Omega\times S,&\\
u_I\in L_+^\infty(\overline{\Omega})\cap H_1,&\\
v_I(x,\cdot)\in L_+^\infty(B(x))\cap H_2\,\,\mbox{for a.e.}\,\,x\in\overline{\Omega}.&
\end{cases}
\end{align*}
\end{assumption}

We also define the following constants for later use:
\begin{align}
&M_1:=\max\{\|u_I\|_{L^\infty(\Omega)},\|u_b\|_{L^\infty(\Omega)}\},\label{m1}\\
&M_2:=\max\{\|v_I\|_{L^\infty(\Omega)},M_1\}.\label{m2}
\end{align}
Note that $M_1$ and $M_2$ depend on the initial and boundary data, but not on the final time $T$.
Let us introduce the evolution triple $(\mathbb{V},\mathbb{H},\mathbb{V}^*)$,
where
\begin{align}
&\mathbb{V}:=\{(\phi,\psi)\in V_1\times V_2\,|\,\phi(x)=\psi(x,y)\, \mbox{for}\,
x\in \Omega,\, y\in \partial B(x)\},\\
&\mathbb{H}:=H_1\times H_2,
\end{align}

Denote $U:=u-u_b$ and notice that $U=0$ at $\partial \Omega$.

\begin{definition}\label{def_weak}
Assume Assumptions 1 and 2. The pair $(u,v)$, with $u=U+u_b$ and
where $(U,v)\in\mathbb{V}$, is a weak solution of the problem $(P)$
if the following identities hold
\begin{align}
&\int_\Omega \theta \partial_t(U+u_b)\phi\, dx+\int_\Omega
(D\nabla_x(U+u_b)-q(U+u_b))\cdot \nabla_x \phi \, dx=\nonumber\\
&\hspace{8cm}-\int_\Omega \int_{\partial B(x)}\nu_y\cdot (D_l\nabla_y v)\phi \, d\sigma
dx,\\
&\int_{\Omega} \int _{B(x)}\partial_t v \psi \,dydx +
\int_{\Omega} \int _{B(x)} D_l \nabla_y \cdot \nabla_y \psi \,dydx =
\int_\Omega \int_{\partial B(x)}\nu_y\cdot (D_l\nabla_y v)\phi \, d\sigma
dx,
\end{align}
for all $(\phi,\psi)\in \mathbb{V}$ and $t\in S$.
\end{definition}

As a last item in this section on the functional framework, we mention
for reader's convenience the following lemma by Lions and Aubin
\cite{Lions}, which we will need later on:
\begin{lemma}(Lions-Aubin)\label{Aubin} Let $B_0 \hookrightarrow B \hookrightarrow B_1$ be Banach spaces such that $B_0$ and $B_1$ are reflexive and the embedding $B_0\hookrightarrow B$ is compact. Fix $p,q>0$ and let $$W=\left\{z\in L^p(S;B_0): \ \frac{dz}{dt}\in L^q(S;B_1)\right\}$$
with
$$||z||_W:=||z||_{L^p(S;B_0)}+||\partial_t z||_{L^q(S;B_1)}.$$
Then $W\hookrightarrow\hookrightarrow L^p(S;B)$.
\end{lemma}

\subsection{Estimates and uniqueness}
%

In this section we establish the positivity and boundedness of the concentrations. Furthermore, we prove an energy inequality and ensure the uniqueness of weak solutions to problem (P).

\begin{lemma}
Let Assumptions 1 and 2 be satisfied. Then any weak solution $(u,v)$ of
problem $(P)$  has the following properties:
\begin{enumerate}
\item[(i)] $u\geq 0$ for a.e. $x\in \Omega$ and for all $t\in S$;
\item[(ii)] $v\geq 0$ for a.e. $(x,y)\in\Omega\times B(x)$ and for all $t\in S$;
\item[(iii)] $u\leq M_1$ for a.e. $x\in \Omega$ and for all $t\in S$;
\item[(iv)] $v\leq M_2$ for a.e. $(x,y)\in\Omega\times B(x)$ and for all $t\in S$;
\item[(v)] The following energy inequality holds:
\begin{eqnarray}
\|u\|^2_{L^2(S;V_1)\cap L^\infty(S;H_1)}&+&\|v\|^2_{L^2(S;L^2(\Omega,V_2))\cap L^\infty(S;H_2)} \nonumber\\
&+& \| \nabla_xu\|^2_{L^2(S;H_1)} +
\|\nabla_y v\|^2_{L^2(S\times\Omega\times B(x))} \leq c_1
\end{eqnarray}
\end{enumerate}
where $M_1$ and $M_2$ are given in \eqref{m1} and \eqref{m2}, and where $c_1$ is a constant independent of $u$ and $v$.
\end{lemma}
\begin{proof}
We  prove (i) and (ii) simultaneously. Similar arguments combined with corresponding suitable choices of test functions  lead in a straightforward manner to (iii), (iv), and (v). We omit the proof details.
Choosing in the weak formulation as test functions $(\varphi,\psi):=(-U^-,-v^-)\in\mathbb{V}$, we obtain:
\begin{eqnarray}\label{nn}
\frac{1}{2}\int_\Omega \phi(\partial_t U^-)^2&+&\frac{1}{2}\int_\Omega\int_{B(x)}\partial_t (v^-)^2+\int_\Omega D|\nabla U^-|^2+\int_\Omega \int_{B(x)}D_\ell |\nabla_y v^-|^2\nonumber\\
&=&\int_\Omega \phi\partial_t u_bU^-+\int_\Omega D\nabla u_b\nabla U^--\int_\Omega \nabla\cdot\left(q(U+u_b)\right)\nabla U^-\nonumber\\
&\leq & \int_\Omega D\nabla u_b\nabla U^- -\int_\Omega q(\nabla U+\nabla u_b)\nabla U^--\int_\Omega (U+u_b) {\rm div} q \nabla U^-\nonumber\\
& = & \min_{\overline{\Omega}}q\int_\Omega |\nabla U^-|^2 +\int_\Omega U^- {\rm div} q\nabla U^-\nonumber\\
 &-&\int_\Omega U^+ {\rm div} q\nabla U^-+\int_\Omega(D\nabla u_b -u_b{\rm div} q)\nabla U^-.
\end{eqnarray}
Note that, excepting the last two terms, the  right-hand side of  (\ref{nn}) has the right sign. Assuming, additionally, a compatibility relation between the data $q, u_b$, for instance, of the type $D\nabla u_b=u_b {\rm div }q$ a.e. in $\Omega\times S$, makes the last term of the r.h.s. of (\ref{nn}) vanish. The key observation in estimating the last by one term is the fact that the sets $\{x\in\Omega: U(x)\geq 0\}$ and $\{x\in\Omega: U(x)\leq 0\}$ are Lebesque measurable. This allow to proceed as follows:
\begin{equation}
\int_\Omega U^+{\rm div q}\nabla U^-=\int_{\{x\in \Omega:U(x)\geq 0\}}U^+{\rm div }q\nabla U^-+\int_{\{x\in \Omega:U(x)\leq 0\}}U^+{\rm div }q\nabla U^-=0.
\end{equation}
\begin{eqnarray}
\end{eqnarray}
After applying the inequality between the arithmetic and geometric means applied to the second term for the right hand-side of  (\ref{nn}), the conclusion of both (i) and (ii) follows via the Gronwall's inequality.
\end{proof}

\begin{proposition}[Uniqueness] Problem ($P$) admits at most one weak solution.
\end{proposition}
\begin{proof}
Let ($u_i,v_i$), with $i\in \{1,2\}$, be two distinct arbitrarily chosen weak solutions. Then for the pair $(\rho,\theta):=(u_2-u_1,v_2-v_1)$ we have
\begin{eqnarray}
\int_\Omega \phi\partial_t\rho\varphi &+&\int_\Omega D\nabla \rho\nabla\varphi -\int_\Omega q \rho\nabla\varphi\nonumber\\
&+&\int_\Omega\int_{B(x)}\partial_t\theta\psi+\int_\Omega\int_{B(x)} D_\ell\nabla_y\theta\nabla_y\psi=0
\end{eqnarray}
for all $(\varphi,\psi)\in\mathbb{V}$.

Choosing now as test functions $(\varphi,\psi):=(\rho,\theta)\in\mathbb{V}$, we reformulate the latter identity as:
\begin{equation}\label{STAR}
\int_\Omega\frac{\phi}{2} (\partial_t\rho)^2+\int_\Omega\int_{B(x)}\frac{1}{2}(\partial_t \theta)^2+\int_\Omega D|\nabla\rho|^2+\int_\Omega\int_{B(x)}D_\ell |\nabla_y\theta|^2=\int_\Omega q\rho \nabla\rho.
\end{equation}
Noticing that for any $\epsilon>0$ we can find a constant $c_\epsilon\in ]0,\infty[$ such that
$$\int_\Omega q\rho\nabla\rho \leq \epsilon \int_\Omega |\nabla \rho|^2+c_\epsilon||q||_\infty^2\int_\Omega |\rho|^2,$$
then (\ref{STAR}) yields:
\begin{eqnarray}\label{2STAR}
\frac{1}{2}\frac{d}{dt}\int_\Omega \phi|\rho|^2&+&\frac{1}{2}\frac{d}{dt}\int_\Omega\int_{B(x)}|\theta|^2+\int_\Omega(D-\epsilon)|\nabla \rho|^2\nonumber\\
&+&\int_\Omega\int_{B(x)}D_\ell |\nabla_y\theta|^2\leq c_\epsilon||q||_\infty^2\int_{\Omega}|\rho|^2.
\end{eqnarray}
Choose
\begin{equation}\label{eps}
\epsilon\in \left]0, \min_{\overline{\Omega\times B(x)}} D\right].
 \end{equation}
 Since for all $x\in\overline\Omega$ and $y\in \overline{B(x)}$ we have $\theta(x,y,0)=\rho(x,0)=0$, (\ref{2STAR}) together with (\ref{eps}) allow for the direct application of Gronwall's inequality. Consequently, the solutions $(u_i,v_i)$ with $i\in\{1,2\}$ must coincide a.e. in space and for all $t\in S$.
\end{proof}

\begin{remark}
At the technical level, the merit of the basic estimates enumerated in this section is that they are derived in the $x$-dependent framework and not in a fixed-domain formulation. Note also that the proof of uniqueness does not rely on the use of $L^\infty$- and positivity estimates on concentrations.
\end{remark}

\subsection{Existence of weak solutions}\label{existence}

In this section, we prove existence of weak solutions of problem
$(P)$. We will do this using the Schauder fixed-point argument. The
operator, for  which we seek a fixed point, maps the space
$L^2(S;L^2(\Omega))$ into itself, and consists of a composition of
three other operators. In order to define these operators, we need
the following functional framework:
\begin{align}
&X_1:= L^2(S;L^2(\Omega)),\\
&X_2:= L^2(S;H^1_0(\Omega))\cap H^1(S;L^2(\Omega)),\\
&X_3:=L^2(S;V_2)\cap H^1(S;L^2(\Omega;L^2(B(x)))).
\end{align}

The first operator $T_1$ maps a $f\in X_1$ to
the solution $w\in X_2$ of
\begin{align}
&\int_\Omega \theta \partial_t(U+u_b)\phi\, dx+\int_\Omega
(D\nabla_x(U+u_b)-q(U+u_b))\cdot \nabla_x \phi \, dx=-\int_\Omega f \phi \, dx,
\end{align}
for all $\phi\in H_0^1(\Omega)$.

The second operator $T_2$ maps a $w \in X_2$ to a solution $v\in X_3$ of
\begin{align}
\int_\Omega \int_{B(x)} \partial_t (V+w)\psi \, dydx+
\int_\Omega \int_{B(x)} D_l \nabla_y(V+w)\cdot\nabla_y\psi \, dydx= \nonumber \\
\int_\Omega \int_{\partial B(x)}\nu_y\cdot (D_l\nabla_y(V+w))\psi \,
d\sigma dx, \label{eqv}
\end{align}
for all $\psi \in V_2$ and $t\in S$.

The third operator $T_3$ maps a $v\in X_3$ to $f\in X_1$ by
\begin{align}
f=\int_{\partial B(x)} \nu_y \cdot \nabla_y v\, d\sigma.
\end{align}

The operator $T: X_1 \rightarrow X_1$ of which a fixed point corresponds to a weak solution
op problem $(P)$ is now given by
\begin{align}
T:=T_3 \circ T_2 \circ T_1.
\end{align}

\begin{lemma}
The operator $T$ is well-defined and continuous.
\end{lemma}
\begin{proof}
Since the auxiliary problem (obtained by fixing $f$) is well-posed (see e.g. chapter 3 in \cite{LSU}), we easily see that $T_1$ is well-defined. Furthermore, by standard arguments we can ensure the stability of the weak solution to the latter problem with respect to initial and boundary data and especially with respect to the choice of the r.h.s. $f$, that is $T_1$ maps continuously $X_1$ into $X_2$.

Analogously, same arguments lead to the well-definedness of $T_2$ and to its continuity from $X_2$ to $\hat{X_2}\subset X_3$. The fact that the linear PDE \eqref{eqv} and its weak solution depend (continuously) on the fixed parameter $x\in \Omega$ is not "disturbing" at this point\footnote{Note however that this $x$-dependence will play a crucial role in getting (at a later stage) the compactness of $T_2$.}.

Since for any $v\in X_3$ the gradient $\nabla_yv$ has a trace on $\partial B(x)$, the well-definedness and continuity of $T_3$ is ensured.
\end{proof}


Furthermore we need for the fixed-point argument that the operator $T$ is
compact. It is enough that one of the operators $T_1$, $T_2$ and $T_3$ is
compact. Here we will show that $T_2$ maps $X_2$ compactly into $X_3$.

\begin{lemma}[Compactness]\label{compact}
The operator $T_3\circ T_2$ is compact.
\end{lemma}
\begin{proof}
We will first reformulate \eqref{eqv} by mapping the $x$-dependent domains for the $y$-coordinate to the referential domain $B(0)$ so that the transformed solution $\hat{v}$ is
in $L^2(S;L^2(\Omega;L^2(B(0))))\cap H^1(S;L^2(\Omega;L^2(B(0))))$

This transformation is a mapping $\Psi: \Omega\times B(0) \rightarrow
\Omega \times B(x)$.
We call $\Psi$ a {\em regular $C^2$-motion} if $\Psi\in C^2(\Omega\times B(0))$
with the property that for each $x\in \Omega$
\begin{align}
\Psi(x,\cdot):B(0)\rightarrow B(x):=\Psi(x,B(0))
\end{align}
is bijective, and if there exist constants $c,C>0$ such that
\begin{align}
c\leq \det \nabla_y\Psi (x,y)\leq C,
\end{align}
for all $(x,y)\in\Omega\times B(0)$.
The existence of such a mapping is ensured by the fact that $S_0\in
C^2(\overline{\Omega\times U})$, by Assumption 1.

If $\Psi$ is a regular $C^2$-motion, then the quantities
\begin{align}
F:=\nabla_y\Psi\,\,\mbox{and}\,\,J:=\det F
\end{align}
are continuous functions of $x$ and $y$. Furthermore, we have the following calculation rules:
\begin{align*}
&\nabla_y v=F^{-T}\nabla_{\hat{y}} \hat{v},\\
&\partial_t v=\partial_t\hat{v},\\
&\int_{\partial B(x)} \nu_y\cdot j \,d\sigma=\int_{\Gamma_0} JF^{-T}\hat{\nu}_{
\hat{y}}\cdot \hat{j} \,d\sigma.
\end{align*}
The transformed version of \eqref{eqv} is now written as:
let $w\in X_2$ be given, find $\hat{V}\in L^2(S;L^2(\Omega;H^1_0(B(0))))\cap H^1(S;L^2(\Omega;L^2(B(0))))$
\begin{align}
\int_\Omega \int_{B(0)} \partial_t (\hat{V}+w)\psi J \, dydx+
\int_\Omega \int_{B(0)} JF_{-1}D_lF^{-T}\nabla_y(\hat{V}+w)\cdot\nabla_y\psi \, dydx= \nonumber \\
\int_\Omega \int_{\Gamma_0}\hat{\nu}_y\cdot (JF^{-1}D_lF^{-T}\nabla_y(\hat{V}+w))\psi \,
d\sigma dx, \label{eqvhat}
\end{align}
for all $\psi \in L^2(\Omega;H^1_0(B(0))) $ and $t\in S$.

Denote by $\Gamma_0$ the boundary of $B(0)$.

\begin{claim}
$\Gamma_0$ is $C^2$.
\end{claim}
\begin{proof}[Proof of claim] The conclusion of the Lemma is a straightforward consequence of the regularity of $S_0$, by Assumption 1.
\end{proof}

\begin{claim}[Interior and boundary $H^2$-regularity]\label{int} Assume Assumptions 1 and 2 and  take $\hat V_I\in L^2(\Omega,H^1(B(0)))$. Then
\begin{equation}
\hat V\in L^2(S;L^2(\Omega;H^2_{loc}(B(0))\cap H^1_0(B(0)))).
\end{equation}
Since $\Gamma_0$ is $C^2$, we have
\begin{equation}
\hat V\in L^2(S;L^2(\Omega;H^2(B(0))\cap H^1_0(B(0)))).
\end{equation}
\end{claim}
\begin{proof}[Proof of claim]
The proof idea follows closely the lines of Theorem 1 and Theorem 4 (cf. \cite{Evans}, sect. 6.3)
\end{proof}

\begin{claim}[Additional two-scale regularity] Assume the hypotheses of Lemma \ref{int} to be satisfied. Then
\begin{equation}
\hat V\in L^2(S;H^1(\Omega;H^2(B(0))\cap H^1_0(B(0)))).
\end{equation}
\end{claim}
\begin{proof}[Proof of claim]
Let us take $\emptyset\neq \Omega'\subset \Omega$ arbitrary such that $h:=dist(\Omega',\partial \Omega)>0$. At this point, we wish to show that
\begin{equation}\label{Reg}
\hat V\in L^2(S;H^1(\Omega';H^2(B(0))\cap H^1_0(B(0)))).
\end{equation}
The extension to $L^2(S;H^1(\Omega;H^2(B(0))\cap H^1_0(B(0))))$ can be done with help of a cutoff function as in \cite{Evans} (see e.g. Theorem 1 in sect. 6.3). We omit this step here and refer the reader to {\em loc. cit.} for more details on the way the cutoff enters the estimates. To simplify the writing of this proof, instead of $\hat V$ (and other functions derived from $\hat V$) we  write $V$ (without the hat). Furthermore, since here we focus on the regularity w.r.t. $x$ of the involved functions, we omit to indicate the dependence of $U$ on $t$ and of $V$ on $y$ and $t$.
For all $t\in S$, $x\in \Omega'$ and $Y\in Y_0$, we denote by $U_h^i$ and $V_h^i$ the following difference quotients with respect to the variable $x$:
\begin{eqnarray}
U_h^i (x,t) & := & \frac{U(x+he_i,t)-U(x,t)}{h},\nonumber\\
V_h^i(x,y,t) &:= & \frac{V(x+he_i, y,t)-V(x, y,t)}{h}\nonumber.
\end{eqnarray}
We have for all $\psi\in L^2(\Omega',H_0^1(B(0)))$ the following identities:
\begin{eqnarray}
&&\int_{\Omega'\times B(0)}J(x+he_i)\partial_t (V(x+h e_i)+U(x+he_i))\psi +  \int_{\Omega'\times B(0)} S(x+h e_i)\nabla_y V(x+he_i)\nabla_y\psi\nonumber\\
&- &\int_{\Omega'\times \Gamma_0}\nu_y\cdot (S(x+he_i)D_\ell\nabla_yV(x+he_i))\psi d\sigma =0
\end{eqnarray}
and
\begin{eqnarray}
&&\int_{\Omega'\times B(0)}J(x)\partial_t (V(x)+U(x))\psi +  \int_{\Omega'\times B(0)} S(x)\nabla_y V(x)\nabla_y\psi\nonumber\\
&- &\int_{\Omega'\times \Gamma_0}\nu_y\cdot (S(x)D_\ell\nabla_yV(x))\psi d\sigma =0.
\end{eqnarray}
Subtracting the latter two equations, dividing the result by $h>0$ and choosing then as test function $\psi:=V_h^i$ yields the expression $$A_1+A_2+A_3=0,$$
where
\begin{eqnarray}
A_1&:=& \int_{\Omega'\times B(0)}V_h^i\left[J(x+he_i)\partial_t(V(x+he_i)+U(x+he_i))-J(x)\partial_t(V(x)+U(x))\right]\frac{1}{h}\nonumber\\
& = & \int_{\Omega'\times B(0)}V_h^i (\partial_t V_h^i+\partial_t U_h^i)J(x)+\int_{\Omega'\times B(0)}(\partial_t V(x+he_i)+\partial_t U(x+he_i))J_h^i(x)V_h^i\nonumber\\
A_2&:=& \int_{\Omega'\times B(0)}\frac{1}{h}\left[S(x+he_i)\nabla_yV(x+he_i)-S(x)\nabla_yV(x)\right]\nabla_y V_h^i \nonumber\\
&=& \int_{\Omega'\times B(0)}S \nabla_y V_h^i\nabla_y V_h^i+ \int_{\Omega'\times B(0)}S_h^i\nabla_y V(x+he_i)\nabla_y V_h^i\nonumber\\
A_3&:=& -\int_{\Omega'\times \Gamma_0}\frac{1}{h}\nabla_y\cdot \left[S(x+he_i)\nabla_yV(x+he_i)-S(x)\nabla_yV(x)\right]V_h^i\nonumber\\
& = & -\int_{\Omega'\times \Gamma_0}\nu_y\cdot (S_h^i\nabla_y V(x+he_i)+S\nabla_y V_h^iV_h^i)\nonumber.
\end{eqnarray}
Re-arranging conveniently the terms, we obtain the following inequality:
\begin{eqnarray}
\frac{1}{2}\int_{\Omega'\times B(0)}(V_h^i)^2|J(x)|& + & \int_{\Omega'\times B(0)}|S(x)|(\nabla_y V_h^i)^2 \leq \int_{\Omega'\times B(0)}|V_h^i\partial_t U_h^iJ(x)|\nonumber\\
&+ & \int_{\Omega'\times B(0)}|(\partial_t V(x+he_i)+\partial_t U(x+he_i))J_h^i(x)V_h^i|\nonumber\\
&+ &\int_{\Omega'\times B(0)}|S_h^i \nabla_y V(x+he_i)\nabla_y V_h^i|\nonumber\\
& + & \int_{\Omega'\times \Gamma_0}|\nu_y\cdot (S\nabla_y V_h^i)V_h^i|+\int_{\Omega'\times \Gamma_0}|\nu_y\cdot (S_h^i\nabla_y V(x+he_i)V_h^i)|\nonumber\\
&=&\sum_{\ell=1}^5 I_\ell.
\end{eqnarray}
To estimate the terms $I_\ell$ we make use of Cauchy-Schwarz and Young inequalities, the inequality between the arithmetic and geometric means, and of the trace inequality. We get
\begin{equation}
|I_1|\leq \frac{||J||^2_{L^\infty(\Omega'\times B(0))}}{2}||V_h^i||_{L^2(\Omega'\times B(0))}+\frac{1}{2}||\partial_t U_h^i||_{L^2(\Omega'\times B(0))},
\end{equation}
\begin{eqnarray}
|I_2|& \leq & \frac{||J||^2_{L^\infty(\Omega'\times B(0))}}{2}2\left(||\partial_t V(x+he_i)||_{L^2(\Omega'\times B(0))}+||\partial_t U(x+he_i)||_{L^2(\Omega'\times B(0))}\right)\nonumber\\
&+& ||V_h^i||_{L^2(\Omega'\times B(0))},
\end{eqnarray}
\begin{equation}
|I_3|\leq \epsilon ||\nabla_y V_h^i||^2_{L^2(\Omega'\times B(0))}+c_\epsilon ||S_h^i||^2_{L^\infty(\Omega'\times B(0))} ||\nabla_y V(x+he_i)||^2_{L^2(\Omega'\times B(0))},
\end{equation}
\begin{eqnarray}
\int_{\Omega'\times \Gamma_0}|\nu_y\cdot (S\nabla_y V_h^i)V_h^i| &\leq &||S||_{L^\infty(\Omega'\times \Gamma_0)} ||V_h^i||_{L^\infty(\Omega'\times \Gamma_0)}\int_{\Omega'\times \Gamma_0}|\nu_y\cdot \nabla_y V_h^i|\nonumber\\
& \leq &  |B(0)|^\frac{1}{2}||S||_{L^\infty(\Omega'\times \Gamma_0)} ||V_h^i||_{L^\infty(\Omega'\times \Gamma_0)}
||V_h^i||^2_{L^1(\Omega';H^2(B(0)))},\nonumber\\
\end{eqnarray}
and
\begin{eqnarray}
\int_{\Omega'\times \Gamma_0}|\nu_y\cdot (S\nabla_y V(x+he_i))V_h^i| &\leq &|B(0)|^\frac{1}{2}||S||_{L^\infty(\Omega'\times \Gamma_0)} ||V_h^i||_{L^\infty(\Omega'\times \Gamma_0)}
||V||^2_{L^1(\Omega';H^2(B(0)))}.\nonumber\\
\end{eqnarray}
Note that all terms $|I_\ell|$ are bounded from above. To get their boundedness we essentially rely on the energy estimates for $V$, $U$, $U_h^i$ as well as on the $L^\infty$-estimates on $V$ and $V_h^i$ on sets like $\Omega'\times B(0)$ and $\Omega'\times \Gamma_0$. The conclusion of this proof follows by applying Gronwall's inequality.
\end{proof}

Using the claims above, we are now able to finish the proof of Lemma \ref{compact}, by noting that $T_3\circ T_2: L^2(S;H^1(\Omega;H^2\cap H_0^1(B_0)))\to L^2(S;H^1(\Omega))$ is continuous and compact via applying Lemma \ref{Aubin} with $B_0=H^1(\Omega)$ and $B=B_1=L^2(\Omega)$.
\end{proof}
Putting now together the above results, we are able to formulate the
main result of section \ref{analysis}, namely:
\begin{theorem}\label{main_result}
Problem (P) admits at least a global-in-time weak solution in the
sense of Definition \ref{def_weak}.
\end{theorem}

\section{Discussion}\label{discussion}
The remaining challenge is to make the asymptotic homogenization
step (the passage  $\epsilon\to 0$) rigorous. Due to the
$x$-dependence of the microstructure the existing rigorous ways of
passing to the limit seem to fail   \cite{BLM,HJM,MM10}. As next step,
we hope to be able to marry succesfully  the philosophy of the
corrector estimates
 analysis by Chechkin and Piatnitski \cite{Chechkin} with the intimate two-scale structure of our model.


\end{document}